\documentclass[abstract,DIV=15]{scrartcl}

\usepackage{packages-arxiv}
\usepackage{commands-arxiv}

\bibliography{bibliography}

\title{New Contraction Bounds for Multidimensional Asymptotic Consensus in Dynamic Networks}

\setkomafont{author}{\normalsize}

\makeatletter
\renewcommand*\@fnsymbol[1]{\the#1}
\makeatother

\author{%
  Zeremonienmeister Berber Aab\thanks{\href{mailto:lorke@tu-berlin.de}{\texttt{lorke@tu-berlin.de}}}, %
  Darya Melnyk\thanks{\href{mailto:melnyk@tu-berlin.de}{\texttt{melnyk@tu-berlin.de}}}, and %
  Stefan Schmid\thanks{\href{mailto:stefan.schmid@tu-berlin.de}{\texttt{stefan.schmid@tu-berlin.de}}}\\[0.35em]
  {\small Technische Universit\"at Berlin, Germany}%
}

\date{}

\hypersetup{
  pdftitle={New Contraction Bounds for Multidimensional Asymptotic Consensus in Dynamic Networks},
  pdfauthor={Zeremonienmeister Berber Aab, Darya Melnyk, Stefan Schmid}
}

\begin{document}
\maketitle
\begin{abstract}
  %
%
This work studies the multidimensional asymptotic agreement problem in dynamic networks. We assume that nodes have $d$-dimensional inputs and need to converge arbitrarily close to each other's inputs while their vectors stay inside the convex hull of all inputs. 

We present a new upper bound of $\sqrt{2d/(3d+1)}$ on the contraction rate of asymptotic agreement in the non-split network model, where $d$ is the dimension of the input and $\sqrt{2d/(3d+1)}\rightarrow \sqrt{2/3}\approx 0.816$ for $d\rightarrow\infty$. To this end, we adapt the BallMidpoint algorithm (Melnyk, 2026) recently introduced for the all-to-all communication model. This algorithm lets the nodes choose the midpoint of the smallest enclosing ball of the received vectors. This bound is strictly worse than the $\sqrt{1/2}\approx 0.707$ contraction rate for fault-tolerant approximate agreement in all-to-all communication networks because the convex hulls of the nodes do not have a common intersection, and only intersect pairwise. Our bound improves over the previously best-known contraction rate of $\sqrt{7/8}\approx 0.935$ for dynamic networks via the MidExtremes algorithm (Függer and Nowak, 2018).

We present the first multi-dimensional lower bound of $2/3$ for the contraction of coordinate-free memoryless anonymous deterministic algorithms in non-split networks, for $d\ge 8$. This result shows that the previously best-known lower bound of $1/2$ on the contraction rate for $n\ge 3$ and $d=1$ is not tight in higher dimensions. If $n$ is infinite, we extend this lower bound to $\sqrt{1/2\cdot d/(d+1)}$, which converges to $\sqrt{1/2}$ for $d\rightarrow \infty$. We further show that for a single round, the one-round contraction is at least $\sqrt{\frac{2d-2}{3d+7}}$, which converges to $\sqrt{2/3}$ for $d\rightarrow\infty$, showing that the BallMidpoint strategy is asymptotically optimal for one round.

\end{abstract}

%
%
\section{Introduction}

In asymptotic agreement, nodes in a distributed system communicate their values (states or locations) to other nodes and update them over time such that all nodes' values converge. This type of agreement problem is of importance in practical systems where exact agreement is not possible, such as in robot gathering~\cite{BOUZID20103154, CFPS03, PATTANAYAK2019145}, or is simply too expensive, such as in collaborative learning~\cite{El-Mhamdi21,CMM25}. In robot gathering, an agreement algorithm is typically run until the nodes' locations are close enough to each other, which corresponds to the approximate agreement problem. In collaborative learning, only a constant number of update steps are executed to make sure that the nodes' models do not diverge too much over time. This problem is called the averaging agreement problem. Both problems can be viewed as a special case of asymptotic agreement, where the nodes have an additional stopping condition.

In many practical applications, the input is multidimensional. Multidimensional asymptotic agreement has been considered in the literature under the additional condition that the value to which the nodes are converging should lie in the convex hull of all input values (convex validity condition). Much of the previous research work has studied how this condition can be satisfied if some of the nodes in the system fail~\cite{MHVG15, VG13,MH13}, under synchronous~\cite{VG13,MHVG15}, asynchronous~\cite{MHVG15, VG13,MH13} and network-agnostic~\cite{GLW23,GMM26} communication models, as well as analogs of convex validity in other domains such as graphs~\cite{NR19,AER23,FGPR26,RV26} and simplices~\cite{RV26,AR26}.
The main focus of this line of work is on impossibility results and the resilience of the algorithms.

In this work, we consider asymptotic agreement in dynamic networks. In dynamic networks~\cite{KLO10,CBS09}, the communication graph connecting the nodes changes over time. This model is motivated by wireless and mobile networks, e.g., drones, where the nodes can move and communicate with other nodes in their vicinity. Through the movement, the communication graph may change, which presents new difficulties for algorithm design compared to static networks. Charron-Bost et al.~\cite{CBFT15} showed that asymptotic agreement is solvable if and only if the communication graph in every round has a rooted spanning tree. We thus restrict the analysis in our work to the rooted tree model, where the communication graph is restricted to be a rooted spanning tree in every round.

The question of contraction of multidimensional asymptotic agreement algorithms was first raised by Függer and Nowak~\cite{FN18}. The contraction is the fraction by which the nodes get closer together in every round of the algorithm. The authors analyze the contractions of their algorithms in the all-to-all static networks and in the dynamic non-split networks. They present the first algorithms with constant dimension-independent contractions for the mentioned models, called the MidExtremes and ApproachExtreme, with contractions $\sqrt{7/8}$ and $\sqrt{31/32}$, respectively. While the first algorithm chooses the midpoint between the furthest two of the received points, the second algorithm chooses the midpoint between the local value of a node and the furthest received point from this value. Recently, the contraction for static all-to-all algorithms has been improved to $\sqrt{1/2}$ in~\cite{M26}. The authors use the BallMidpoint strategy for convergence, where each node locally computes the smallest enclosing ball of a so-called safe area, and chooses the midpoint of this ball as the input for the next round.

In this work, we apply the BallMidpoint strategy to show a new upper bound on the contraction in dynamic networks. Since the dynamic communication model is more general than all-to-all communication, we show that $\sqrt{\frac{2d}{3d+1}}$ is the best one-round contraction that we can achieve with this strategy. Similar to~\cite{FN18}, the presented algorithm is coordinate-free, i.e., it does not rely on a fixed coordinate system; and it is memoryless, in the sense that the nodes only use the values from the same round for decision. Moreover, our algorithm is anonymous, as the BallMidpoint strategy always has a unique result and we thus do not need additional ID-based tie-breaking.  We further present the first multi-dimensional lower bounds on the contractions for dynamic networks. To date, the best-known lower bound for non-split networks was $1/2$, and it was derived for the one-dimensional case.

\paragraph*{Our contribution.}
The main contributions of this paper are summarized as follows:
\begin{itemize}
    \item We present the first coordinate-free memoryless anonymous asymptotic agreement protocol with a contraction rate of at most $\sqrt{\frac{2d}{3d+1}}$ for dynamic non-split networks. The contraction rate improves over the previous work for $d\ge4$ (Centroid~\cite{CBFN16multidim} has a better contraction rate for $d\le3$).  We show that the presented contraction analysis is tight for one round. We also show that this strategy converges in the rooted tree model, but the algorithm is not memoryless. 
    \item We present a lower bound of $2/3$ on the contraction for $d\ge 8$ for coordinate-free memoryless anonymous algorithms, which holds for $n \ge 241$ nodes.
    \item We present a lower bound of $\sqrt{\frac{d}{2(d+1)}}$ for coordinate-free memoryless anonymous algorithms, with a countably infinite number of nodes. We further extend this lower bound to show that, if the nodes do not know $n$, and given a number of rounds $R$, there is a construction where any coordinate-free memoryless anonymous algorithm has a contraction of at least $\sqrt{\frac{d}{2(d+1)}}$ in every round.
    \item We show that for one single round of communication, no coordinate-free memoryless anonymous algorithm can have a contraction better than $\sqrt{\frac{2d-2}{3d+7}} \xrightarrow{d \to \infty} \sqrt{\frac 2 3}$, showing that the analysis of our algorithm is asymptotically optimal for one round.
\end{itemize}

The presented lower bounds are derived in non-split networks, where any two nodes in the communication graph have a common in-neighbor. The lower bounds of $2/3$ and $\sqrt{\frac{d}{2(d+1)}}$ are constructed such that every node receives points whose convex hull is a regular $d$-simplex. Due to the symmetry property of a regular $d$-simplex, any coordinate-free memoryless anonymous algorithm can select only one point as the input for the next round, namely the centroid of the convex hull. This way, we can construct examples where the inputs of the nodes look identical to the inputs of the previous round, up to rotation and scaling. Note that our assumption of coordinate-free memoryless anonymous algorithms does not extend to the MidExtremes algorithm proposed by Függer and Nowak~\cite{FN18}. This is because MidExtremes in some cases requires tie-breaking, e.g., if a node receives an equilateral triangle. To break ties, the algorithm either needs a coordinate system, or some type of node IDs, or randomization. Thus, the algorithm cannot be deterministic, anonymous, and coordinate-free at the same time.

To reach a lower bound of $\sqrt{\frac{2}{3}}$ for one round asymptotically as $d \to \infty$, we consider non-split graphs where any two views of nodes only intersect in one point. Instead of using regular $d$-simplices, we use products of two regular simplices to construct a shape that has $d$ symmetry axes. Due to the symmetry property, the centroid is the only point that any coordinate-free memoryless anonymous algorithm can choose. We are not able to repeat this construction for several rounds, but it shows that the presented algorithm analysis is asymptotically optimal for one round.

\paragraph*{Paper outline.}
In \Cref{sec:reltedWork}, we present the related work. We start the main part of the paper by formalizing the asymptotic agreement problem and different dynamic network models used in this work in Section~\ref{sec:model}. In Section~\ref{sec:AAAlgorithm}, we present the BallMidpoint algorithm in the non-split network model. We analyze its contraction in Section~\ref{sec:BallMidpoint}, the tightness of the analysis in Section~\ref{sec:algAnalysisTightness}, and the generalization to rooted tree networks in Section~\ref{sec:rootedTrees}. We then turn to lower bounds in Section~\ref{sec:LBs}, where we first present the $2/3$ lower bound for $d\ge 8$ with $241$ nodes (Section~\ref{sec:LBEightDim}). We then improve this lower bound to $\sqrt{\frac{d}{2(d+1)}}$ under the assumption that the number of nodes is infinite or for a bounded number of rounds (Section~\ref{sec:LBInfinite}). Finally, in Section~\ref{sec:LBOneRound}, we show that the contraction of any coordinate-free memoryless anonymous deterministic algorithm in non-split networks in a single communication round is lower bounded by $\sqrt{\frac{2d-2}{3d+7}} \xrightarrow{d \to \infty} \sqrt{2/3}$.

\section{Related Work}\label{sec:reltedWork}

\paragraph{Asymptotic and approximate agreement.} Asymptotic agreement algorithms have been predominantly studied on static communication graphs. The problem was first introduced in one dimension by Dolev et al.~\cite{DLPSW86} to overcome the FLP impossibility result~\cite{FLP85} for consensus under asynchronous communication with one fault. In one-dimensional approximate agreement, the goal for the nodes is to converge to a value inside the range of all non-faulty input values. Mendes and Herlihy~\cite{MH13} and Vaidya and Garg~\cite{VG13} in parallel generalized this result to multiple dimensions. Their focus was on finding possibility and impossibility results for multidimensional approximate agreement under convex validity, i.e., the nodes need to converge inside the convex hull of all non-faulty input vectors. Their work has sparked interest in multidimensional agreement that has today been considered in different communication models~\cite{MHVG15,GLW23,GMM26}, on different underlying domains~\cite{NR19,AER23,FGPR26,RV26, AR26}, for applications such as optimization~\cite{LGV21} and federated learning~\cite{El-Mhamdi21}, as well as under relaxed validity conditions~\cite{XV17,CAMBUS2026116044, CMM25, CMMS26, CMMS25, MS26}.

\paragraph{Dynamic networks and contractions.} Asymptotic agreement in dynamic networks was first considered by Charron-Bost et al.~\cite{CBFT15}. The authors showed that asymptotic agreement in dynamic networks only works if the communication graph in each round contains a rooted spanning tree.
They provide a proof showing that weighted averaging algorithms converge in this model. They also show that any product of $n-1$ rooted graphs with $n$ nodes is non-split, that is, an asymptotic agreement algorithm in a non-split network can be easily generalized to rooted graphs.
The asymptotic agreement algorithms have been extended to the multidimensional case in~\cite{CBFN16multidim}, where the Centroid of the convex hull of the received values is selected as the new input. This algorithm has been shown to achieve a contraction of $d/(d+1)$. Függer and Nowak~\cite{FN18} improved this bound to $\sqrt{7/8}$ by presenting the MidExtremes algorithm, which chooses the midpoint between the two furthest received points. As their focus is on efficient local computation, they also present the ApproachExtreme strategy, where each node chooses the midpoint between its own value and the furthest received value. The faster local computation comes at a cost of a worse contraction of $\sqrt{31/32}$. This contraction bound has been recently improved to $\sqrt{1/2}$ for static all-to-all communication algorithms using the BallMidpoint~\cite{M26}. The BallMidpoint strategy lets each node locally compute the smallest enclosing ball that contains the so-called safe area, and then outputs the midpoint of the ball. The analysis of the contraction in this paper relies on the assumption that all locally received convex hulls have a non-empty intersection. Such an assumption cannot be guaranteed in dynamic networks. Our paper adapts the BallMidpoint algorithm to dynamic networks and shows a new upper bound on the contraction in this model.

\paragraph{Lower bounds on contractions.} Fekete~\cite{F86} analyzed contraction bounds in static all-to-all networks under various network failures in one dimension, presenting asymptotically tight bounds.
Függer et al.~\cite{FNS21} show the first lower bounds for dynamic network models subject to an oblivious message adversary. In particular, they show that the contraction rate of $1/3$ is tight for graphs with $2$ nodes. For $n\ge 3$ and $d=1$, they show that the contraction of $1/2$ is tight for dynamic networks which include deaf graphs. This bound also holds for non-split networks. For $n\ge 4$, they show that the contraction under a rooted adversary is at least $\sqrt[n-2]{1/2}$ in one round. Compared to non-split networks, the rooted network model provides no guarantee that the views of the nodes intersect in one round; only after $n-1$ rounds can the nodes be guaranteed to have intersecting views.
In contrast to previous work, our work presents the first multidimensional lower bounds for non-split graphs that improve over the one-dimensional $1/2$ contraction bound.

\section{Model and Definitions}\label{sec:model}

We consider dynamic networks on finitely many nodes\footnote{In Section~\ref{sec:LBInfinite}, we drop the assumption and consider countably infinite sets of nodes.}. We assume that each node holds a state, which is updated over time, taking values in $\values$, a $d$-dimensional $\R$-vector space equipped with an inner product $\map{\innprod \cdot \cdot}{\values \times \values} \R$, meaning it is bilinear, symmetric, and positive-definite. Communication graphs in a dynamic network can change over time. In particular, we assume that the communication graphs are chosen adversarially.

In this work, we consider deterministic algorithms that solve asymptotic agreement.  We assume that the algorithms are \emph{memoryless}, that is, the only memory that a node can hold between communication phases, which we call its \emph{state}, is a single value from the set of all possible values that the nodes want to solve consensus for.  This state is the only data a node can send to other nodes, and in turn it only receives such states from other nodes. Additionally, we assume that the algorithms are \emph{anonymous}, that is, the nodes do not know their own ID, nor from which nodes they received inputs.  Furthermore, the execution of the algorithm shall be \emph{coordinate-free}, i.e., equivariant under similarity maps: translation, scaling, and orthogonal linear maps\footnote{that is, linear maps $\map L \values \values$ preserving the inner product, meaning $\innprod x y = \innprod{\apply L x}{\apply L y}$ for all $x,y \in \values$}.

\paragraph{Adversary.} We define the adversary as a class of communication graphs $\adversary$, which assigns to each possible set of nodes $\nodes$ a subset $\apply \adversary \nodes \subseteq \digraphs \nodes$ of the set of all simple directed graphs on $\nodes$ without loops.
For a set of nodes $\nodes$, the adversary may choose any sequence $\Gamma = \tuple{G_i}_{i \in \N_{>0}} \in \apply \adversary \nodes ^{\N_{>0}}$, which will influence the execution of each node in the network.  We furthermore define $\adversary^\ast \defeq \bigunion_\nodes \apply \adversary \nodes ^{\N_{>0}}$.
In this paper, $\nodes$ ranges over all finite sets.
We consider two dynamic network models. In the \emph{rooted tree model}, each communication graph is a rooted spanning tree. In the non-split network model, we define a set of non-split graphs $\nonsplit \nodes$ on the set of nodes $\nodes$, which contains graphs whose closed in-neighborhoods have pairwise nonempty intersection. In the \emph{non-split network model}, we assume that the adversary chooses non-split graphs as the set of communication graphs $\adversary = \tuple{\nonsplit \nodes}_{\nodes}$. We assume that the communication is synchronous, i.e., it is round-based. Each adversarially chosen graph $G_i$ is associated with one round $\round \in \N_{>0}$ of the synchronous communication.

\paragraph{Algorithm execution.} To define the execution of an algorithm, for a node $v \in \nodes$ and a graph $G \in \digraphs \nodes$ we first define the set of in-neighbors $\innbhd G v \defeq \set{w \in \nodes \given \tuple{w,v} \in \edges G}$ as well as the closed in-neighborhood of $v$ as $\cinnbhd G v \defeq \innbhd G v \union \set v$, which is its in-neighborhood including itself. Let $\tuple{q_v}_{v \in \nodes} \in \values^\nodes$ denote a configuration of states of the nodes.  We define the multiset of received inputs of $v$ for the communication graph $G$ as $\received{v}{G}{q} = \multiset{q_w}_{w \in \cinnbhd G v}$.

Given an adversary $\adversary$, an \emph{instance} of this model is defined by a tuple $J = \tuple{\transition,I,\Gamma}$, consisting of a transition function $\map \transition {\multi \values} {\values}$ (which defines the algorithm), a configuration of initial values $I = \tuple{I_v}_{v \in \nodes} \in \values^\nodes$, and a sequence of graphs $\Gamma = \tuple{G_i}_{i \in \N_{>0}} \in \apply \adversary \nodes ^{\N_{>0}} \subseteq \adversary^{\ast}$ for some set of nodes $\nodes$.  To satisfy being coordinate-free, we require that the transition function $\map \transition {\multi \values} {\values}$ satisfies $\apply L { \apply \transition M} = {\apply \transition {\apply L M}} = \apply \transition {\multiset{\apply L m}_{m \in M}}$ for all \emph{similarity maps} $\map L \values \values$, i.e., compositions of translations, scalings, and linear orthogonal maps.

We define the execution of an instance as the sequence of states for each node that result from locally updating the state by applying the transition function to all of the states in the closed in-neighborhood of that node according to the current graph.  More precisely, for an instance $J = \tuple{\transition,I,\Gamma}$ for nodes $\nodes$, we define the states $\state J r v$ for $v \in \nodes$ and $r \in \N_0$ recursively via
\begin{itemize}
  \item $\state J 0 v = I_v$
  \item $\state J r v = \apply \transition {\received{v}{G_r}{\stater J {r-1}}}$.
\end{itemize}
Observe the anonymous nature of this process: The transition function takes a multiset of values as input, which remains the same if the same states come from different nodes.

\paragraph{Asymptotic agreement.}\label{def:asymptoticAgreement} Using an execution, we can now define the asymptotic agreement problem, which we want to solve in this paper: 

\begin{definition}[Asymptotic agreement~\cite{FN18}]
    An algorithm solves the \emph{asymptotic agreement} problem in a given dynamic network model if, for every instance $J$, the following properties are satisfied:  
    \begin{itemize}
        \item \emph{Convergence:} For every node $v\in V$, the sequence $(\state J r v)_{r\in\N_0}$ converges.
        \item \emph{Agreement:} If $(\state J r v)_{r\in\N_0}$ and $(\state J r w)_{r\in\N_0}$ converge, they converge to the same state.
        \item \emph{Validity:} If $(\state J r v)_{r\in\N_0}$ converges, then its limit is inside the convex hull of all initial states $I$.
    \end{itemize}
\end{definition}

We now present a general coordinate-free memoryless anonymous deterministic algorithm that we later use to solve asymptotic agreement. In each round, every node executes \Cref{alg:node}, where $\map \transition {\multi \values} \values$ is the (abstract) transition function of the algorithm. In each round $r \in \N_{>0}$, a node $v$ starts with state $q_v' = \state J {r-1} v$, sends this value to $O_v = \outnbhd {G_r} v$, receives $M_v = \received v {G_r} {q'_v} = \set{\state J {r-1} w \given w \in \cinnbhd{G_r}{v}}$, and outputs $q_v = \apply \transition {M_v} = \state J r v$, in accordance with the mathematical definition of our model.

\begin{algorithm}
  \caption{one round at node $v \in \nodes$}
  \label{alg:node}
  \DontPrintSemicolon
  \KwIn{old state $q_v' \in \values$ of $v$; recipients $O_v \subseteq \nodes$; received points $M_v \subseteq \values$}
  \KwOut{new state $q_v \in \values$}
  send $q_v'$ to all $w \in O_v$\;
  receive the set of points $M_v$\;
  $q_v \gets \apply \transition {M_v}$\;
  \Return{$q_v$}
\end{algorithm}

\paragraph{Contraction.} We evaluate the convergence speed of asymptotic agreement algorithms using the contraction, which we define similarly to~\cite{M26}. We first define the smallest enclosing ball of a set of points: 
\begin{definition}[Smallest enclosing ball]\label{def:smallestEncBall}
  Let $M \in \multi\values$ be a multiset of points in $\values$.  We define its \emph{circumcenter} as $\cntr M \in \values$ and its \emph{circumradius} $\radius M \in \R_{\geq 0}$ as the unique center and radius of a ball such that $M \subseteq \minball M \defeq \ball {\cntr M} {\radius M}$ where $\radius M$ is minimized, and we call $\minball M$ the \emph{smallest enclosing ball} of $M$.
\end{definition}

Note that the smallest enclosing ball is always unique~\cite{jung1901ueber}. In the following, we use the smallest enclosing ball to define the contraction in an algorithm in one round:

\begin{definition}[One-round contraction]
    For two bounded sets $A,B\subseteq \values$, we define the contraction from $A$ to $B$ as $\contraction A B \defeq \radius B / \radius A$, the ratio of the radius of the smallest enclosing ball of $B$ to the one of $A$.  If $\radius A = 0$, all states coincide and equivariance forces this point as output (\Cref{lem:symcentroid}), so $\radius B = 0$.  We then define $\contraction A B = 0$.

    For an instance $J = \tuple{\transition, I, \Gamma}$, let $M_{r-1} = \set{\state J {r-1} v}_{v \in \nodes}$ be the states of the nodes at the beginning of round $r$ and $M_r = \set{\state J r v }_{v \in \nodes}$ be the new states of the nodes outputted at the end of round $r$ using \Cref{alg:node}.  Then, the contraction $\roundcontraction J r$ of instance $J$ in round $r$ is defined as $\contraction {M_{r-1}} {M_r}$, the contraction from $M_{r-1}$ to $M_r$.

    The contraction rate of instance $J$ is then defined as $\crate J \defeq \limsup_{k \to \infty} \sqrt[k]{\prod_{r=1}^k \roundcontraction J r}$.  For a given transition function $\transition$, we define its contraction rate as $\cratet \transition \defeq \sup_{I \in \values^\nodes \commaand \Gamma \in \adversary^\ast} \crate {\tuple{\transition,I,\Gamma}}$, the supremum over all contraction rates of all instances with transition function~$\transition$.
\end{definition}

In our analysis, we will focus on the contraction for the non-split network model only. In this model, the convex hulls of the received states of any two local views of the nodes have a common intersection, which allows us to show smaller contraction rates. Note that there is a reduction from rooted trees to non-split networks, showing that similar overlaps of views can also be achieved in the rooted tree model~\cite{CBFT15}. Note that this reduction requires $n-1$ rounds, and it is not memoryless.

\section{Asymptotic Agreement Algorithm}\label{sec:AAAlgorithm}

In this section, we present the BallMidpoint algorithm to solve asymptotic agreement in dynamic networks. We show that it converges in the rooted tree model and present an upper bound on the contraction rate of this algorithm in the non-split network model. Additionally, we show that the analysis of the upper bound is exactly tight for one round.

The following definition presents the \bmp algorithm. As our setting is synchronous and memoryless, an algorithm is completely determined by a transition function $\map \transition {\multi \values} \values$:

\begin{definition}[\bmp]\label{def:bmp}
  We define the \bmp algorithm as \Cref{alg:node} together with the transition function $\map \tbmp {\multi \values} \values$ where $M \mapsto \cntr M$. Here, $\cntr M$ is the circumcenter of the smallest enclosing ball of $M$ as in \Cref{def:smallestEncBall}.
\end{definition}

Note that \bmp is coordinate-free:  balls are equivariant under similarities ($\apply L {\ball c \rho} = \ball{\apply L c}{s\rho}$, where $s$ is the scaling factor of $L$), hence $\cntr {\apply L M} = \apply L {\cntr M}$.

\subsection{Contraction of BallMidpoint in non-split networks}\label{sec:BallMidpoint}

In this section, we prove the contraction of the \bmp algorithm (as defined in \Cref{def:bmp}) in non-split networks:

\begin{theorem}\label{thm:BMPupperBound}
   The \bmp algorithm solves asymptotic agreement in non-split networks, and its contraction rate is upper bounded by $\sqrt{\frac {2d} {3d+1}}$.
\end{theorem}

We will prove the upper bound on the contraction rate by upper bounding the contraction in every round of the algorithm. Before proving this theorem, we present some important results that help derive the full proof.

In the following, let $\tuple{p_v}_{v \in \nodes} \subseteq \values$ be a set of points and $\tuple{B_v}_{v \in \nodes}$ be a set of balls in $\values$ such that for all $v \in \nodes$ we have $B_v = \ball {p_v} {\rho_v}$ for some $\rho_v \in \R_{\geq 0}$.

\begin{lemma}\label{lem:intersect}
  Let $\tuple{B_v}_{v \in \nodes}$ be a collection of balls that pairwise intersect.  Then for $\tuple{B'_v}_{v \in \nodes}$, where $B'_v = \ball {p_v} {\rho'_v}$ with $\rho'_v = \sqrt{\frac{2d}{d+1}} \rho_v$ for $v \in \nodes$ we have $\bigintersect_{v \in \nodes} B'_v \neq \emptyset$.
\end{lemma}
\begin{proof}
    This result follows from Theorem 5.4 in~\cite{CMMS26}. This paper uses the Soddy-Gosset theorem~\cite{Sod37,SoddyGossetOtherProofEnglish} to prove that scaling the radius of each of $d+1$ mutually tangent hyperspheres in $d$ dimensions by a factor of $1+\frac{d-1}{d+1+\sqrt{2(d+1)d}} = \sqrt{\frac{2d}{d+1}}$ around its center gives all hyperspheres a common intersection point. 
\end{proof}

We next show that the circumcenter of the smallest enclosing ball of a set of points lies in the convex hull of the set of points. This lemma helps us prove the upper bound on the contraction and later show that the validity condition from \Cref{def:asymptoticAgreement} is satisfied.

\begin{lemma}\label{lem:centerhull}
  Let $Q \subseteq \values$ be a nonempty bounded set of points.  Then $\cntr Q \in \closure{\chull{Q}}$.
\end{lemma}

\begin{proof}
  Let $c \in \values \setminus \closure{\chull Q}$.  Since $\closure{\chull Q}$ is closed and convex, there is a unique $p \in \closure{\chull Q}$ minimizing $\norm{c-p}$, which satisfies $\innprod{q-p}{c-p} \le 0$ for all $q \in Q$.  Hence, with $\delta \defeq \norm{p-c} > 0$, \[
    \norm{q - c}^2 = \norm{q - p}^2 - 2 \innprod{q - p}{c - p} + \norm{p - c}^2 \ge \norm{q - p}^2 + \delta^2
  \] for all $q \in Q$.  Thus $\ball{p}{\sqrt{R^2 - \delta^2}} \supseteq Q$ for $R \defeq \sup_{q \in Q} \norm{q - c}$, so no ball containing $Q$ with center outside $\closure{\chull Q}$ can have minimum radius, i.e., $\cntr Q \in \closure{\chull Q}$.
\end{proof}

The following is a well-known result that any $d$-dimensional smallest enclosing ball can be defined by at most $d+1$ points~\cite{jung1901ueber}: 

\begin{lemma}\label{lem:activepoints}
  Let $Q \subseteq \values$ be a nonempty bounded set of points.  Then there exists a set of \emph{active points} $A \subseteq \closure Q \intersect \bndry{\minball Q}$ with $\cardinality A \le d+1$ such that $\minball A = \minball Q$.
\end{lemma}

To prove the contraction rate, we upper bound the distance from any possible circumcenter to a selected point inside the smallest enclosing ball of all points. To this end, we use the non-split property of the considered networks, which guarantees that any two locally received sets have a common point. The following two lemmas provide vectors inside the local convex hull that we use to later bound the distances. Note that the proof is similar to the proof in~\cite{M26}. However, since we cannot assume that all locally received sets have a common intersection point, we compare the distances to a point chosen with respect to the intersection of the increased balls according to \Cref{lem:intersect}. 

\begin{lemma}\label{lem:positivescalar}
  Let $Q \subseteq \values$ be nonempty and bounded, and let $u \in \values$ be any vector.  Then there exists a point $a \in \closure Q$ such that $$\innprod{a - \cntr Q}{u} \ge 0.$$
\end{lemma}

\begin{proof}
  For the halfspace $H_{\cntr Q,u} = \set{x \in \values \given \innprod{x - \cntr Q}{u} < 0}$, we have $\cntr Q \notin H_{\cntr Q,u}$, since $\innprod{\cntr Q - \cntr Q}{u} = 0$.  Since $\cntr Q \in \closure{\chull Q} \subseteq \chull{\closure Q}$ with $\chull{\closure Q}$ closed, the convexity of $H_{\cntr Q,u}$ yields that $\closure Q$ is not contained in $H_{\cntr Q,u}$, as otherwise its convex hull would be.  Hence there is $a \in \closure Q$ with $a \notin H_{\cntr Q,u}$.
\end{proof}

\begin{lemma}\label{lem:mainlemma}
  Let $P \subseteq \values$ be a bounded set of points and let $M \subseteq P$.  We then have
  \begin{equation}
    {\radius P}^2 \ge {\radius M}^2 + \norm{\cntr M - \cntr P}^2 \mathperiod
  \end{equation}
\end{lemma}

\begin{proof}
  Let $A$ be as in \Cref{lem:activepoints}.  In particular, $\cntr A = \cntr M$, and with \Cref{lem:positivescalar}, there exists $a \in A$ such that $\innprod{a - \cntr M}{\cntr M - \cntr P} \ge 0$.  Since $a \in \bndry{\minball M}$, we have $\norm{a - \cntr M} = \radius M$.  Note also that $a \in \closure P \subseteq \minball P$, and therefore $\norm{a - \cntr P} \leq \radius P$.  We can then calculate
  \begin{align*}
    {\radius P}^2
    &\ge \norm{a - \cntr P}^2 = \innprod {a - \cntr P} {a - \cntr P} \\
    &= \innprod {a - \cntr M + \cntr M - \cntr P} {a - \cntr M + \cntr M - \cntr P} \\
    &= \norm{a - \cntr M}^2 + 2 \innprod {a - \cntr M}{\cntr M - \cntr P} + \norm{\cntr M - \cntr P}^2 \\
    &\ge \norm{a - \cntr M}^2 + \norm{\cntr M - \cntr P}^2 = \radius M ^2 + \norm{\cntr M - \cntr P}^2 \mathperiod     \qedhere
  \end{align*}
\end{proof}

We now present the proof for the main theorem of this section:

\begin{proof}[Proof of \Cref{thm:BMPupperBound}]
  Let $\tuple{p_i}_{i \in \nodes} \subseteq \values$ be any bounded collection of points, and let $G \in \nonsplit \nodes$ be any non-split graph.  Let $P \defeq \set{p_i}_{i \in \nodes}$ and $M_i \defeq \set{p_j \given j \in \cinnbhd G i}$.   Let $\nu \defeq \frac{d+1}{3d+1}$, noting that $1-\nu = \frac{2d}{3d+1}$ and $\frac{1-\nu}{\nu} = \frac{2d}{d+1}$.  Now consider a point $q \in \bigintersect_{i \in \nodes} \ball {\cntr{M_i}}{\sqrt{\frac{1-\nu}{\nu}} \radius{M_i}} \neq \emptyset$, which is a point inside the intersections of increased balls as in \Cref{lem:intersect}.  Such a point exists by \Cref{lem:intersect}:  As $G$ is non-split, the $M_i$ pairwise intersect, and hence so do the balls $\ball{\cntr{M_i}}{\radius{M_i}} \supseteq M_i$.
  
  Now let $x \defeq \cntr P + \nu(q - \cntr P) = (1-\nu)\cntr P + \nu q$, and let $u \defeq \cntr{M_i} - \cntr P$ and $v \defeq \cntr{M_i} - q$.  We then have $\norm{v} \leq \sqrt{\frac{1-\nu}{\nu}}\radius{M_i}$, as well as $\norm{u} \leq \sqrt{\radius{P}^2 - \radius{M_i}^2}$ via \Cref{lem:mainlemma}.  We can now upper bound the distance from any locally computed circumcenter $\cntr{M_i}$ to a designated point~$x$
  \begin{align*}
    \norm{\cntr{M_i} - x}^2
    &= \norm{(1-\nu)u + \nu v}^2 
    \le (1-\nu)\norm{u}^2 + \nu\norm{v}^2 \\
    &\le (1-\nu)\bigl(\radius{P}^2 - \radius{M_i}^2\bigr) + \nu\mult \frac{1-\nu}{\nu}\radius{M_i}^2  
    \frac{2d}{3d+1}\radius{P}^2 \mathcomma
  \end{align*}
  using convexity of the squared norm in the first inequality.
  Taking square roots yields $\norm{\cntr{M_i} - x} \le \sqrt{\frac{2d}{3d+1}}\,\radius{P}$.

  Finally, we show correctness of the \bmp algorithm.  By \Cref{lem:centerhull}, the state a node computes lies in the convex hull of the states it received, and hence in the convex hull of all states of the previous round.  The convex hulls of the state sets of successive rounds are therefore nested.  From the analysis above, the states after round $r$ are contained in a ball whose radius shrinks by the factor $\sqrt{\frac{2d}{3d+1}}$ in every round, so the diameter of the nested hulls tends to $0$.  The nested hulls thus intersect in a unique common point, and since each node's state lies in the hull of its round, the states of all nodes converge to this point.  By nesting, this point lies in the convex hull of the inputs, establishing validity.  This proves that the \bmp algorithm solves asymptotic agreement.
\end{proof}

\subsection{Tightness of the contraction analysis for one round}\label{sec:algAnalysisTightness}

In this section, we present a construction that forces the contraction of the \bmp algorithm to be $\sqrt{\frac{2d}{3d+1}}$ in one round, showing that our previous analysis was tight. The ratio of the new radius $\sqrt{d/(d+1)}$ and the initial radius $\sqrt{(3d+1)/(2(d+1))}$ matches the bound presented in \Cref{thm:BMPupperBound} exactly, so the analysis is tight for every dimension. To prove the lower bound, we will present a point distribution for which the contraction can be lower bounded. One way to realize such a point distribution is to assign the points as inputs to the nodes of the algorithm.

\begin{lemma}\label{lem:maxnorm}
  Let $Q \subseteq \values$ be a finite set of points and let $c \in \chull Q$.  If $R \defeq \norm{q-c} = \norm{q' -c}$ for all $q,q' \in Q$, then $\cntr Q = c$ and $\radius Q = R$.
\end{lemma}

\begin{proof}
  Via assumptions, we have $\ball c R \supseteq Q$.  Since $c \in \chull Q$, we may write $c = \sum_{q \in Q} a_q q$ with $0 \le a_q \le 1$ and $\sum_{q \in Q} a_q = 1$.  Let now $c' \in \values$ and $R' \in \R_{\geq 0}$ with $\ball {c'} {R'} \supseteq Q$.  Then $\sum_{q \in Q} a_q \norm{q-c'}^2 = \sum_{q \in Q} a_q \norm{q-c}^2 - 2\innprod{\sum_{q \in Q} a_q q - c}{c' - c} + \norm{c - c'}^2 = R^2 + \norm{c - c'}^2 \ge R^2$.  Hence there exists $q \in Q$ with $\norm{q-c'} \ge R$, therefore $R' \ge R$.
\end{proof}

\begin{theorem}
  Let $d \ge 2$.  There is an instance such that the contraction of the \bmp algorithm for one round is $\sqrt{\frac{2d}{3d+1}}$.
\end{theorem}

\begin{proof}
  Consider points $\set{u_i}_{i = 0}^d \subseteq \values$ that form a regular $d$-simplex with edge length $\sqrt 2$ with the origin as its centroid, implying $\innprod{u_i}{u_i} = \frac{d}{d+1}$ and $\innprod{u_i}{u_j} = -\frac{1}{d+1}$ for all distinct $i,j \in \set{0,\dotsc,d}$.  Now consider the points $p_{i,j} = \frac 1 2 (u_i + u_j)$ for $i,j \in \set{0,\dotsc,d}$, $j<i$ and the points $p_{i,-1} = u_i - \frac 1 2 (u_{i-1 \mod (d+1)} - u_{i+1 \mod (d+1)})$ and $p_{i,-2} = u_i + \frac 1 2 (u_{i-1 \mod (d+1)} - u_{i+1 \mod (d+1)})$ for $i \in \set{0,\dotsc,d}$.

  To define $G$, we first define a set of nodes $\nodes = \set{\tuple{i,j} \in \set{0,\dotsc,d} \times \set{-2,-1,0,\dotsc,d} \given j < i}$.  We next declare the state of each node $\tuple{i,j} \in \nodes$ to be $s_{\tuple{i,j}} = p_{i,j}$.  For each node $\tuple {i,j} \in \nodes$, we define $N_i = \set{\tuple{i,k}}_{k \in \set{-2,-1,0,\dotsc,i-1}} \union \set{\tuple{k,i}}_{k \in \set{i+1,\dotsc,d}}$ and we define a communication graph defined via $\cinnbhd G {\tuple{i,j}} = N_i$ for each $\tuple {i,j} \in \nodes$.  For $S_i = \multiset{s_v}_{v \in N_i}$, we have $\received {\tuple{i,j}} G {s} = S_i$ for all $\tuple{i,j} \in \nodes$.
  
  We first show that $G$ defined as above is non-split. Note that for any two nodes $\tuple{i_1,j_1}, \tuple{i_2,j_2} \in \nodes$, if $i_1 = i_2$, we have $\cinnbhd G {\tuple{i_1,j_1}} = \cinnbhd G {\tuple{i_2,j_2}}$, and otherwise, assuming $i_2 < i_1$ w.l.o.g., we have $\tuple{i_1,i_2} \in \cinnbhd G {\tuple{i_1,j_1}} \intersect \cinnbhd G {\tuple{i_2,j_2}}$.  This also implies that $\multiset{S_v}_{v\in \nodes}$ are pairwise intersecting.

  Next, we show that \bmp chooses $u_i$ as the new points, i.e., $\cntr {S_i} = u_i$.  To see this, note first that the radius of the smallest ball with circumcenter $u_i$ that contains $S_i$ is $\sqrt{1/2}$:
  \begin{align*}
    \norm{u_i - p_{i,-2}}
    &= \norm[\Big]{- \frac 1 2 (u_{i-1 \mod (d+1)} - u_{i+1 \mod (d+1)})}\\
    &= \frac 1 {\sqrt 2} = \norm{u_i - p_{i,-1}} \mathperiod
  \end{align*}
  Also, for any $0\le j<i\le d$, and similarly for $0\le j<i<k\le d$ we have $\norm{u_i - p_{i,j}} = \norm{u_i - p_{k,i}} = \frac 1 {\sqrt 2}$.  Since $p_{i,-2}$ and $p_{i,-1}$ are antipodal with respect to $u_i$, no smaller radius than $\frac 1 {\sqrt 2}$ is possible, which shows $u_i = \cntr {S_i}$.

  It remains to show the contraction by bounding the radius of the initial smallest enclosing ball and the newly computed one.  Let $M_\ast \defeq \multiset{s_v}_{v \in \nodes}$ be the multiset of initial states and $M_{\ast+1} \defeq \multiset{\apply \transition {S_i}}_{\tuple{i,j} \in \nodes}$ the multiset of states after one round. We start by determining the radius of the smallest enclosing ball of the initial values. For all $i,j,k \in \set{0,\dotsc,d}$ with $j<i$, we have $\norm{p_{i,j}} = \sqrt{\frac{d-1}{2(d+1)}} \le \sqrt{\frac{3d+1}{2(d+1)}} = \norm{p_{k,-2}} = \norm{p_{k,-1}}$.  The points of maximal norm are $p_{k,-1}$ and $p_{k,-2}$, whose midpoint is $u_k$; since the $u_k$ average to the origin, \Cref{lem:maxnorm} shows $\radius {M_\ast} = \sqrt{\frac{3d+1}{2(d+1)}}$.  For the newly computed points, we can compute the radius exactly: For all $\tuple{i,j} \in \nodes$, we have $\norm{\apply \transition {S_i}} = \norm{u_i} = \sqrt{\frac{d}{d+1}}$, which implies $\radius{M_{\ast+1}} = \sqrt{\frac{d}{d+1}}$.  This shows $\contraction {M_\ast}{M_{\ast+1}} = \radius{M_{\ast+1}}/\radius{M_\ast} = \sqrt{\frac{2d}{3d+1}}$.
\end{proof}

\subsection{\bmp in rooted tree networks}\label{sec:rootedTrees}

In the following, we state that \bmp converges in the rooted tree model:

\begin{theorem}
    The \bmp algorithm can be used to solve asymptotic agreement in the rooted tree model. 
\end{theorem}
\begin{proof}
    To prove this theorem, we first use the reduction from rooted trees to non-split networks from~\cite{CBFT15}, showing that a product of $n-1$ rooted trees forms a non-split graph (\cite{CBFT15}, Lemma 6). The proof of the lemma provides translation of the rooted tree model into a macro-round with a non-split communication graph. By applying \bmp from \Cref{def:bmp} at the end of each macro-round, we receive an asymptotic agreement algorithm in rooted tree networks.
\end{proof}

\section{Lower Bounds on Contraction in Non-Split Networks}\label{sec:LBs}

In this section, we show lower bounds on the contraction of coordinate-free memoryless anonymous deterministic algorithms in non-split networks. We start by proving a bound of $2/3$ for every round of an algorithm using a constant number of nodes in $d\ge 8$ dimensions. This lower bound is constructed in such a way that the initial inputs build a symmetric structure, which after one round is reached again up to rotation and a scaling factor. We show that the scaling factor is $2/3$ in every round for every coordinate-free memoryless anonymous deterministic algorithm. We then present a similar construction that uses infinitely many nodes, while providing a lower bound of $\sqrt{\frac{d}{2(d+1)}}$ on the contraction for every round of an algorithm. We also show how this construction can be achieved if the number of nodes is finite and the number of rounds is fixed for an algorithm. Additionally, we show that the lower bound for the contraction within one round converges to $\sqrt{2/3}$ for $d \to \infty$. This shows that the \bmp algorithm has an asymptotically optimal contraction for one round.

Our lower bound constructions rely on the fact that any transition function is forced to map \enquote{sufficiently symmetric} sets to a unique point. This property allows us to determine which point any coordinate-free memoryless anonymous deterministic algorithm will choose for a given example. For example, for a hypercube, a regular simplex, or a regular polytope, we have sufficiently many symmetries such that the transition function must map its defining corner points to its barycenter/center of mass. Next, we will formally prove this statement.

\begin{lemma}\label{lem:symcentroid}
    Let $M \in \multi \values$ be a finite multiset.  Consider its affine symmetry group $\apply {\operatorname{Sym}} M$, defined as the group of all transformations $L$ of $\values$ that are a composition of translations and orthogonal linear transformations, which leave $M$ invariant, meaning $M = \apply L M = \multiset{\apply L p \given p \in M}$.  Now consider the set of fixed points $\values^{\apply {\operatorname{Sym}} M} = \set{p \in \values \given \forall L \in \apply {\operatorname{Sym}} M \qcolon \apply L p = p}$.  Then any transition function $\transition$ must map $M$ to a point in $\values^{\apply {\operatorname{Sym}} M}$.
\end{lemma}

\begin{proof}
    In our model, any transition function $\transition$ is equivariant under $\apply {\operatorname{Sym}} M$.  Therefore, for any $L \in \apply {\operatorname{Sym}} M$, we have $\apply L {\apply \transition M} = \apply \transition {\apply L M} = \apply \transition M$.  By definition, this means $\apply \transition M \in \values^{\apply {\operatorname{Sym}} M}$.
\end{proof}

As an example, consider a multiset of points that have a rotation as a symmetry.  Then all points on the axis of rotation are fixed points, and any transition function must map to this rotation axis.  Similarly, for a reflection symmetry, it must map to the hyperplane of reflection.  Furthermore, if we have enough symmetries such that these axes and hyperplanes intersect in a unique point, then this point must be the output of the transition function.

\begin{definition}[centroid]
    If $\values^{\apply {\operatorname{Sym}} M} = \set{\cntr M}$ is a singleton for a multiset $M \in \multi \values$, we call $\cntr M$ the \emph{centroid} of $M$. 
\end{definition}

The abuse of notation conflating the centroid and the circumcenter is justified by the fact that they coincide:  Any symmetry $L$ of $M$ maps $\minball M$ to an enclosing ball of $\apply L M = M$ of the same radius, so by uniqueness of the smallest enclosing ball, $\apply L {\minball M} = \minball M$, and hence $\apply L {\cntr M} = \cntr M$.

\subsection{Lower bound for \texorpdfstring{$d \geq 8$}{d >= 8}}\label{sec:LBEightDim}

 In this section, we show that no algorithm can have a contraction better than $2/3$ for $d \geq 8$. We start with an example that shows that in dimension $2$ we can achieve a lower bound of $\sqrt{\frac 1 3}$. 

 \begin{example}
 Consider the set of nodes $\nodes = \set{0,1,2,3,4,5,6}$ that start with a configuration $\state J 0 6 = 0$ and $\state J 0 k = \expo{\tfrac k 6 2\pi \imaginary}$ for $k \in \set{0,1,2,3,4,5}$. Let the non-split graph be defined by $\cinnbhd G 6 = \set 6$ and $\cinnbhd G k = \set{k,k+1 \mod 6, 6}$.  Node $6$ must choose $\state J 0 6 = 0$ as its next state, as this is the only input it receives. All other nodes receive an equilateral triangle, a regular $2$-simplex, as their input. The equilateral triangle is invariant under rotation around its centroid, so a coordinate-free memoryless anonymous deterministic algorithm cannot single out any other point and must choose the centroid as its next state (\Cref{lem:symcentroid}). For $S_r = \set{\state J r k}_{k \in \nodes}$, we have $S_1 = \set{\sqrt{\tfrac 1 3}\expo{\tfrac 1 {12} 2\pi \imaginary} s \given s \in S_0}$. For the input values, the six states of unit norm average to the origin, so \Cref{lem:maxnorm} gives $\radius {S_0} = 1$, and since multiplication with $\expo{\tfrac 1 {12} 2\pi \imaginary}$ is a rotation, we must have $\radius {S 1} = \sqrt{\tfrac 1 3}$. Since the construction for the new round is equivalent up to rotation and scaling, $\sqrt{\tfrac 1 3}$ is a lower bound on the contraction rate of any coordinate-free memoryless anonymous deterministic algorithm. One round of the presented construction is visualized in \Cref{fig:lowerbound2}.

\begin{figure}
  \centering
  \begin{tikzpicture}[scale=2.2]

    \foreach \k in {0,...,5}{
      \coordinate (v\k) at ({60*\k}:1);
    }
    \coordinate (v6) at (0,0);

    \foreach \k in {0,...,5}{
      \pgfmathtruncatemacro{\knext}{mod(\k+1,6)}
      \draw[navyblue] (v\k) -- (v\knext) -- (v6) -- cycle;
    }

    \foreach \k in {0,...,5}{
      \fill[navyblue] (v\k) circle (1pt);
      \node[navyblue] at ({60*\k}:1.12) {$\k$};
    }

    \foreach \k in {0,...,5}{
      \pgfmathtruncatemacro{\knext}{mod(\k+1,6)}
      \coordinate (c\k) at (barycentric cs:v\k=1,v\knext=1,v6=1);
    }

    \foreach \k in {0,...,5}{
      \draw[black, dotted, thick, -{Latex[length=2.5mm]}, shorten >=2pt] (v\k) -- (c\k);
    }

    \foreach \k in {0,...,5}{
      \pgfmathtruncatemacro{\knext}{mod(\k+1,6)}
      \draw[newcol] (c\k) -- (c\knext) -- (v6) -- cycle;
    }

    \foreach \k in {0,...,5}{
      \fill[newcol] (c\k) circle (0.8pt);
      \node[newcol] at ({30+60*\k}:0.68) {$\k$};
    }

    \fill[navyblue] (v6) circle (1pt);
    \node[navyblue, above,xshift=-3pt,yshift=5.5pt] at (v6) {$6$};
  \end{tikzpicture}
  \caption{Multi-round  $\sqrt{\frac 1 3}$ lower bound for any coordinate-free memoryless anonymous algorithm. }\label{fig:lowerbound2}
  \end{figure}
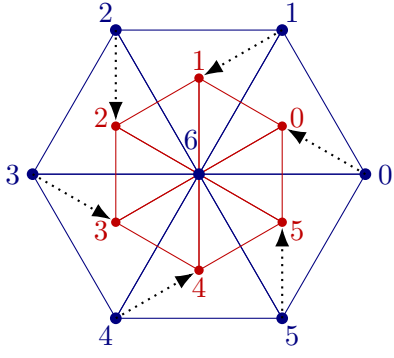

\end{example}

This construction does not simply extend to higher dimensions. In general, the idea is that in dimension $d$, we form regular unit length $(d-1)$-simplices via points on the unit $(d-1)$-sphere. Then each such $(d-1)$-simplex together with the origin forms a regular unit $d$-simplex. Then, when an algorithm chooses the centroids of each such simplex as the new point, it arrives at the same configuration of points, but multiplied by an element of the orthogonal group (rotations, reflections) and scaled by a factor of $\sqrt{\frac{d}{2(d+1)}}$.  This is exactly what happens in our example for $d=2$, but is not possible for $d=3$ with finitely many points. It is not clear if other constructions work to show a general contraction lower bound of $\sqrt{\frac{d}{2(d+1)}}$ for dimension $d$, but we can construct a similar example in dimension $8$ via octonions with a contraction lower bound of $\sqrt{\frac{8}{2(8+1)}} = \frac 2 3$ for $d \geq 8$.

\begin{lemma}\label{lem:DeightLBconstruction}
  Let $\nodes = E_8 \union \set 0$.  For the non-split graph $G$ defined by $\cinnbhd G c = S^c_0$ for $c \in E_8$ and $\cinnbhd G 0 = \set 0$, with the $S^c_0$ constructed below, any transition function results in the set of all updated values being isometric to $\frac 2 3 E_8 \union \set 0$.
\end{lemma}

\begin{proof}
To show our bound in dimension $d=8$, we construct the $E_8$ root system in a real inner product space $\values$ of dimension $8$, as a set of points on the unit sphere $\sphere 7 \subseteq \values$.  Fix an orthonormal basis $\tuple{\basis i}_{i = 1}^8$.  We explicitly construct $E_8$ as the set of all points of the form
\begin{itemize}
\item $\frac{1}{\sqrt 2} \mult (\pm \basis i \pm \basis j)$ for integers $1 \leq i < j \leq 8$, and
\item $\frac 1 {\sqrt 2} \mult \sum_{i=1}^8 \pm \frac 1 2 \basis i$ with an even number of $-$ signs in the sum.
\end{itemize}
These can be verified to have unit length, and we have $\card{E_8} = 4 \binom 8 2 + 2^7 = 240$ points.

Now we need a normed division algebra on $\values$ with respect to the inner product of our vector space $\values$, but in such a way that $E_8$ is closed under multiplication, and such that the unity element is in $E_8$.  This is octonion multiplication, and the existence of this with our requirements on $E_8$ can be found in \cite{Co46} and \cite[Chapter 9]{CS03}.  Furthermore, since the Weyl group of the root system $E_8$ acts transitively and (by definition) isometrically on $E_8$~\cite{Hu72}, we may without loss of generality choose any element in $E_8$ to be the unity, by precomposing the multiplication with such an isometry ($x \omult y = \apply \varphi {\apply {\inverse \varphi} x  \omult' \apply {\inverse \varphi} y }$, where $\varphi$ maps the unity element w.r.t. $ \omult'$ to our desired element).

In total, this amounts to the following:  We have an $\R$-bilinear operation $\map \omult {\values \times \values} \values$ such that:
\begin{enumerate}
\item For all $x,y \in \values$ we have $\norm{x \omult y} = \norm x \norm y$, hence $x \mapsto c \omult x$ and $x \mapsto x \omult c$ for $c \in \values$ are linear isometries whenever $\norm c = 1$, and there are no zero-divisors.
\item $b_1 \defeq \frac 1 {\sqrt 2} \sum_{i=1}^8 \frac 1 2 \basis i$ is a two-sided identity element with respect to $\omult$.
\item $E_8$ is closed under $\omult$.
\end{enumerate}

We form a base clique $B \defeq \set{b_1,\dotsc,b_8} \subseteq E_8$ defined by $b_1 = \frac 1 {\sqrt 2} \sum_{i=1}^8 \frac 1 2 \basis i$ and $b_j = \frac 1 {\sqrt 2} (\basis 1 + \basis j)$ for $2\leq j \leq 8$.  We can verify that they each have unit distance to each other, so they form a regular $7$-simplex, and together with the origin a regular $8$-simplex $B_0 \defeq B \union \set 0$.  For each $c \in E_8$, we define $S^c = \set{c \omult b_j}_{j=1}^8$, which also each form a regular $7$-simplex and $S^c_0 = S^c \union \set 0$ a regular $8$-simplex, as multiplication by $c$ is an isometry.  Note that $c \in S^c$, since $b_1$ is the unity.

By \Cref{lem:symcentroid}, any transition function maps the regular simplex $S^c_0$ to its centroid.  Computing the centroid now simply reduces to computing the centroid of $S_0^{b_1} = B_0$.  Let $\hat g = \frac 1 6 \sum_{i=1}^8 b_i$.  Note that $\norm{\sum_{i=1}^8 b_i} = 6$, so $\hat g$ is a unit vector.  We then calculate
\begin{equation*}
  \apply{\transition}{S^c_0} = \frac 1 9 ( (\sum_{i=1}^8 c \omult b_i) + 0 ) = \frac 6 9 (c \omult \hat g) = \frac 2 3 (c \omult \hat g)
\end{equation*}
The map $\map C {E_8} \values$,  $c \mapsto  \apply{\transition}{S^c_0} = \frac 2 3 (c \omult \hat g)$, is therefore just the composition of an isometry, the multiplication with $\hat g$ from the right, and afterwards a scaling of $\frac 2 3$.  The image of $C$ therefore is just another copy of $E_8$, and we can do the same construction recursively. This shows the lemma statement.
\end{proof}

The construction from \Cref{lem:DeightLBconstruction} embeds into any higher dimension $d > 8$, as a multiset with a centroid keeps the same centroid when embedded.  For $n > 241$, we pad the construction with extra nodes $v$ of state $0$ and closed in-neighborhood $\set{v,0}$, which stay at $0$ and keep the graph non-split.  This implies the following theorem.

\begin{theorem}
  The contraction rate of any coordinate-free memoryless anonymous deterministic asymptotic agreement algorithm in dimension $d \geq 8$ with $n\ge 241$ nodes in the non-split communication model has a lower bound of $\frac 2 3$.
\end{theorem}

\subsection{Lower bound for \texorpdfstring{$\cardinality \nodes = \infty$}{|N| = infinity}}\label{sec:LBInfinite}

In this section, we prove a lower bound of $\sqrt{\frac{d}{2(d+1)}}$ for the contraction rate assuming that we have countably infinitely many nodes.  We take care to not use the axiom of choice in this section.  In our construction, the graphs are locally finite.  Each node only has a finite number of in-neighbors.

To set up our countably infinite construction, consider the real algebraic numbers $K \defeq \overline \Q \intersect R$, or the real closure of $\Q$, where $\overline \Q$ is the algebraic closure of $\Q$.  Let $\values' \defeq \apply {\spann_K}{\set{e_1,\dotsc,e_d}} \subseteq \values$, for an orthonormal basis $\set{e_1,\dotsc,e_d} \subseteq \values$, be the set of algebraic points in $\values$. Next, we define the unit sphere $S_{d-1} \defeq \sphere{d-1} \intersect \values'$ of algebraic points. 

The idea is similar to our construction for $d=2$ and $d=8$:  Take points on the unit sphere, each has an in-neighborhood of a regular simplex, and then we update to their centroids.  To arrive at a self similar structure for general $d$, we take not finitely many points, but all points on $S_{d-1}$, and assign simplices according to the following lemma, illustrated in \Cref{fig:infinite-3d}.

\begin{figure}[t]
\centering
\begin{tikzpicture}[
    x={(-1.2679cm,-0.9299cm)},
    y={( 2.7189cm,-0.4336cm)},
    z={( 0.45cm, 2.5cm)}]
  \def\R{3cm}
  \def\kR{1.8371cm}
  \def\az{25}

  \fill[nodecol!6] (0,0,0) circle[radius=\R];
  \draw[nodecol, thick] (0,0,0) circle[radius=\R];
  \draw[nodecol!60, domain=\az-90:\az+90, samples=60, smooth]
    plot ({cos(\x)},{sin(\x)},0);
  \draw[nodecol!60, dashed, domain=\az+90:\az+270, samples=60, smooth]
    plot ({cos(\x)},{sin(\x)},0);

  \draw[newcol, thick, dashed] (0,0,0) circle[radius=\kR];

  \newcommand{\sample}[7]{%
    \filldraw[simplex] #1 -- #2 -- #3 -- cycle;
    \draw[spoke] (0,0,0) -- #1;
    \draw[spoke] (0,0,0) -- #2;
    \draw[spoke] (0,0,0) -- #3;
    \draw[hray] (0,0,0) -- #4;
    \node[hpt, label={[newcol]right:$h(#6)$}] at #4 {};
    \node[oldpt] at #1 {};
    \node[oldpt] at #2 {};
    \node[oldpt] at #3 {};
    \node[newpt, label={[newcol, xshift=7pt, yshift=2pt]below:$\kappa\,h(#6)$}] at #5 {};
    \node[nodecol, #7] at #1 {$#6$};
  }
  \sample{(-0.3534,-0.4556,0.8165)}{(0.5720,-0.0779,0.8167)}{(-0.2190,0.5335,0.8166)}
         {(0,0,1)}{(0,0,0.6124)}{x}{left}

  \node[origin, label={[yshift=1pt, xshift=2pt]below:$0$}] at (0,0,0) {};

  \node[nodecol] at ($(0,0,0)+(-2.55cm,-2.3cm)$) {$\mathcal S^{2}$};
  \node[newcol]  at ($(0,0,0)+(-1.8cm,-1.8cm)$) {$\kappa\,\mathcal S^{2}$};
\end{tikzpicture}
\caption{Visualization of \Cref{thm:lower} for $d=3$, with $V=\mathcal S^{2}\cup\{0\}$ and $\kappa=\sqrt{3/8}$. Every point of the unit sphere is a node (blue). A node $x$ receives $\sigma(x)$ (shaded) together with $0$ and moves to the centroid $\kappa\,h(x)$ (red), where $h(x)$ (hollow) satisfies $\langle x,h(x)\rangle=\sqrt{2/3}$. The new states form the sphere of radius $\kappa$ (dashed).}

\label{fig:infinite-3d}
\end{figure}
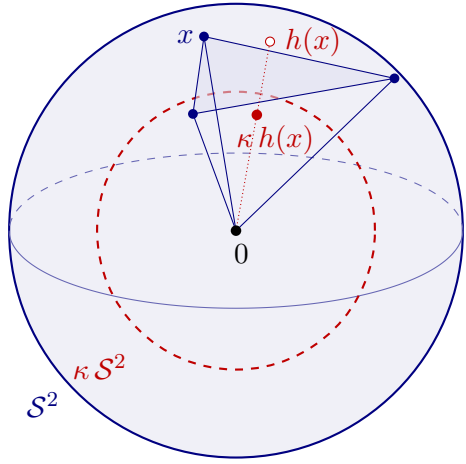

\begin{lemma}\label{lem:bijection}
  Let $d \geq 3$.  For any $\tau \in K$ with $0 < \tau < 1$, there is a bijection $\map h {S_{d-1}} {S_{d-1}}$ such that $\innprod x {\apply h x} = \tau$ for all $x \in S_{d-1}$.
\end{lemma}

\begin{proof}
  We first fix an enumeration (a bijection) $\map {\operatorname{enum}} {\N_0} {S_{d-1}}$.  For $x \in S_{d-1}$, define $I_x = \set{i \in \N_0 \given \innprod x {\apply {\operatorname{enum}} i} = \tau}$, each of which is countably infinite for $d \ge 3$, since $\set{y \in S_{d-1} \given \innprod x y = \tau}$ is a $(d-2)$-sphere in which the algebraic points are dense.  We now inductively define a family of sets $\tuple{P_n}_{n \in \N_0}$ and a family of maps $\tuple{h_n}_{n \in \N_0} \subseteq (P_n \to P_n)$.
  \begin{itemize}
  \item $P_0 = \emptyset$, and $h_0$ is the empty map.
  \item Suppose that $\map {h_n} {P_n} {P_n}$ is a bijection for some $n \in \N_0$.  Define $x_{n+1} = \apply {\operatorname{enum}} {\min (\N_0 \setminus \apply {\inverse {\operatorname{enum}}} {P_n})}$ and $y_{n+1} = \apply {\operatorname{enum}} {\min (I_{x_{n+1}} \setminus \apply {\inverse {\operatorname{enum}}} {P_n})}$.  We define $h_{n+1}$ by extending $h_n$ to include $\set{x_{n+1},y_{n+1}}$, mapping $x_{n+1} \mapsto y_{n+1}$ and $y_{n+1} \mapsto x_{n+1}$.  Explicitly, this means $P_{n+1} = P_n \union \set{x_{n+1},y_{n+1}}$, and $\map {h_{n+1}} {P_{n+1}} {P_{n+1}}$ is defined by \[z \mapsto \apply {h_{n+1}} z = \begin{cases} \apply {h_n} z & z \in P_n, \\ y_{n+1} & z = x_{n+1}, \\ x_{n+1} & z = y_{n+1}. \end{cases}\]  This is then again a bijection.
  \end{itemize}
  
  With this construction, we have $\bigunion_{n \in \N_0} P_n = S_{d-1}$.  To see this, define $J_n = \min (\N_0 \setminus \apply {\inverse {\operatorname{enum}}} {P_{n-1}})$ for $n \in \N_{>0}$.  Then via construction of $P_n$, $J_{n+1} > J_n$ holds for all $n \in \N_{>0}$.  With $J_1 = 0$, we inductively get $J_n \geq n-1$ and therefore $\apply {\operatorname{enum}} n \in P_{n+1}$ for all $n \in \N_0$.  This then shows $z \in P_{\apply{\inverse {\operatorname{enum}}} {z} + 1}$ for any $z \in S_{d-1}$, concluding $\bigunion_{n \in \N_0} P_n = S_{d-1}$.
  
  We can define $\map{h \defeq \bigunion_{n \in \N_0} h_n} {\bigunion_{n \in \N_0} P_n} {\bigunion_{n \in \N_0} P_n}$ (since each $h_{n+1}$ extends $h_n$).  This is the bijection $\map h {S_{d-1}} {S_{d-1}}$, satisfying $\innprod {\apply h x} x = \tau $ for all $x \in S_{d-1}$.
\end{proof}

\begin{lemma}\label{lem:nochoice}
  Let $\tau \in K$ with $0 < \tau < 1$.  There exists a map $\map \Psi {\set{\tuple{x,z} \in S_{d-1} \times S_{d-1} \given \innprod{x}{z} = \tau}} {\set{M \subseteq S_{d-1} \given \card M = d}}$, such that $\apply \Psi {x,z} \subseteq \set{y \in S_{d-1} \given \innprod y z = \tau}$ is a regular $(d-1)$-simplex with $x \in \apply \Psi {x,z}$.

  Moreover, for any transition function $\transition$, we have $\frac 1 \tau \apply \transition {\apply \Psi {x,z}} = z$, and the edge length of $\apply \Psi {x,z}$ is $\sqrt{\frac{(1-\tau^2)2d}{d-1}}$ for all $\tuple{x,z} \in S_{d-1} \times S_{d-1}$ with $\innprod{x}{z} = \tau$.
\end{lemma}

\begin{proof}
  Let $r = \sqrt{1-\tau^2}$, and define $c = \tau z$ and $u_0 = \frac {x-c}{r}$.  Then $u_0 \in S_{d-1}$ with $\innprod{u_0}{z} = 0$.  Note that $K$ is closed under square roots of positive elements, as required by Gram-Schmidt.  Applying an algebraic Householder reflection mapping a vertex of a fixed algebraic regular $(d-1)$-simplex in $\orthcomp z$ to $u_0$ gives a centered regular simplex $\set{u_i}_{i=0}^{d-1} \subseteq \orthcomp z$ of circumradius $1$, with entries in $K$, containing $u_0$.  We can then define $\apply \Psi {x,z} = \set{c+ru_i}_{i=0}^{d-1}$, which fulfills our desired properties.

  Let $\tuple{x,z} \in \set{\tuple{x,z} \in S_{d-1} \times S_{d-1} \given \innprod{x}{z} = \tau}$.  For every $v \in \apply \Psi {x,z}$, we have $\norm{v-\tau z}^2 = 1+ \tau^2 - 2\tau \innprod v z = 1 - \tau^2$.  As $\tau z$ lies in the affine hyperplane $\set{y \in \values \given \innprod y z = \tau}$ spanned by $\apply \Psi {x,z}$, the circumcenter of $\apply \Psi {x,z}$ is $\tau z$, which is equal to its centroid, and we have circumradius $\sqrt{1-\tau^2}$.

  With the formula $\rho = L \sqrt{\frac{d-1}{2((d-1)+1)}}$ for the circumradius $\rho$ and the edge length $L$ of a $(d-1)$-dimensional regular simplex, any transition function must update to its centroid, which by \Cref{lem:symcentroid} is forced to be $\tau z$.  The edge length of $\apply \Psi {x,z}$ then is $\sqrt{\frac{(1-\tau^2)2d}{d-1}}$.
\end{proof}

\begin{lemma}\label{lem:lower}
  Let $d \ge 3$.  There exists a map $\map \sigma {S_{d-1}} {\set{M \subseteq S_{d-1} \given \card M = d}}$ such that $\apply \sigma x$ is a regular $(d-1)$-simplex of unit edge length with $x \in \apply \sigma x$ for all $x \in S_{d-1}$, and such that $\map g {S_{d-1}} {S_{d-1}}$ defined by $x \mapsto \sqrt{\frac{2(d+1)}{d}}\cntr {\apply \sigma x \union \set 0}$ is bijective.
\end{lemma}

\begin{proof}
    Let $\map h {S_{d-1}} {S_{d-1}}$ as in \Cref{lem:bijection} with $\tau =  \sqrt{\frac{d+1}{2d}}$.  Now for each $x \in S_{d-1}$, define a regular $(d-1)$-simplex $\apply \sigma x \subseteq \set{y \in S_{d-1} \given \innprod y {\apply h x} = \sqrt{\frac{d+1}{2d}}}$ such that $x \in \apply \sigma x$.  We can construct $\map \sigma {S_{d-1}} {\set{M \subseteq S_{d-1} \given \card M = d}}$ explicitly without the axiom of choice via $x \mapsto \apply \sigma x \defeq \apply \Psi {x,\apply h x}$ from \Cref{lem:nochoice}.  Also from \Cref{lem:nochoice}, for all $x \in S_{d-1}$, we get a unit edge length $d$-simplex $\apply \sigma x \union \set 0$, and therefore $\apply g x = \sqrt{\frac{2(d+1)}{d}}\cntr {\apply \sigma x \union \set 0} = \apply h x$.  Finally, since $h$ is a bijection, $\sigma$ and $g$ fulfill all our desired properties.
\end{proof}

\begin{theorem}\label{thm:lower}
    Let $d \ge 3$.  If we assume $\nodes$ to be allowed to be countably infinite ($\cardinality\nodes = \cardinality \N$), for any transition function $\transition$, there is an instance whose contraction is lower bounded by $\sqrt{\frac{d}{2(d+1)}}$ in every round, implying that the contraction rate $\cratet \transition$ has lower bound $\sqrt{\frac{d}{2(d+1)}}$.  We do not need the axiom of choice.
\end{theorem}

\begin{proof}
    Let $\nodes = S_{d-1} \union \set 0$.  Consider $\map \sigma {S_{d-1}} {\set{M \subseteq S_{d-1} \given \card M = d}}$ and $\map g {S_{d-1}} {S_{d-1}}$ as in \Cref{lem:lower}.  We define for each $r \in \N_{>0}$ a communication graph $G_r$ on $\nodes$ via $\cinnbhd {G_r} v = \apply {g^{-(r-1)}} {\apply \sigma {\apply {g^{r-1}} v}} \union \set 0$ for $v \in S_{d-1}$ and $\cinnbhd {G_r} 0 = \set 0$, denoting $\Gamma = \tuple{G_r}_{r \in \N_{>0}}$.  These graphs are well-defined since $g$ is a bijection, and non-split since every in-neighborhood contains $0$.  The in-neighborhoods are closed since $x \in \apply \sigma x$ for all $x \in S_{d-1}$ implies $v = \apply {g^{-(r-1)}} {\apply {g^{r-1}} v} \in \cinnbhd {G_r} v$.  Finally, we define initial states $I = \tuple{I_v}_{v \in \nodes}$ defined by $I_v = v$ for $v \in \nodes = S_{d-1} \union \set 0$.  Then, for the instance $J = \tuple{\transition,I,\Gamma}$ for any transition function $\transition$, the first round will have a node $v \in S_{d-1}$ update its state $\state J 0 v = v \in S_{d-1}$ to $\cntr {\apply \sigma v \union \set 0} = \sqrt{\frac{d}{2(d+1)}} \apply g v$, as $\apply \sigma v \union \set 0$ is a regular $d$-simplex (\Cref{lem:symcentroid}).  Since $g$ is a bijection, the set of all states after round one is then $\sqrt{\frac{d}{2(d+1)}} S_{d-1}$.  Inductively, entering round $r$, the state of each node $w \in S_{d-1}$ is $\state J {r-1} w = \sqrt{\frac{d}{2(d+1)}}^{r-1} \apply {g^{r-1}} w$, so node $v \in S_{d-1}$ receives the multiset $\set{\state J {r-1} w \given w \in \cinnbhd {G_r} v} = \sqrt{\frac{d}{2(d+1)}}^{r-1} (\apply \sigma {\apply {g^{r-1}} v} \union \set 0)$, since $g^{r-1}$ is a bijection.  This is again a regular $d$-simplex, scaled by $\sqrt{\frac{d}{2(d+1)}}^{r-1}$, so node $v$ updates its state to its centroid $\sqrt{\frac{d}{2(d+1)}}^{r-1} \cntr {\apply \sigma {\apply {g^{r-1}} v} \union \set 0} = \sqrt{\frac{d}{2(d+1)}}^r \apply {g^r} v$, and we get $\set{\state J r v}_{v \in \nodes} = \sqrt{\frac{d}{2(d+1)}}^r S_{d-1}\union \set 0$.  We therefore have $\roundcontraction J r = \radius {\sqrt{\frac{d}{2(d+1)}}^r S_{d-1}} / \radius {\sqrt{\frac{d}{2(d+1)}}^{r-1}S_{d-1}} = \sqrt{\frac{d}{2(d+1)}}$ for all $r \in \N_{>0}$ via construction.
\end{proof}

The above construction does not work on finitely many points.  However, taking the instance in \Cref{thm:lower}, restricting it to a sufficiently large finite subset of the nodes, and simplifying the remaining in-neighborhoods gives a contraction lower bound for any fixed number of rounds.

\begin{corollary}
    Let $d \ge 3$.  Assuming a model with finitely many nodes in a network, let $\transition$ be any transition function.  For any number of rounds $R \in \N_{>0}$, there is an instance whose contraction is lower bounded by $\sqrt{\frac{d}{2(d+1)}}$ in every round $1 \le r \le R$.
\end{corollary}

\begin{proof}
    Let $\transition$ be any transition function.  Consider an instance $J = \tuple{\transition,I,\Gamma}$ with nodes $\nodes$ as in \Cref{thm:lower}.  With \Cref{lem:activepoints}, choose a subset $A_R \subseteq \nodes$ such that we have $\radius {\set{\state J R v}_{v \in \nodes}} = \radius {\set{\state J R v}_{v \in A_R}}$.  Intuitively, we construct a set $Z$ of nodes obtained by successively applying the closed in-neighborhood on each node in rounds $1$ to $R$, but backwards, also adding sufficiently many active points.  We recursively define, starting at $r=R$ and going backwards to $r=1$, the nodes $Z_R = A_R$ and $Z_{r-1} = \cinnbhd {G_{r}} {Z_r} \union A_{r-1}$, where $A_{r-1} \subseteq \nodes$ is again chosen via \Cref{lem:activepoints} such that $\radius {\tuple{\state J {r-1} v}_{v \in \nodes}} = \radius {\tuple{\state J {r-1} v}_{v \in A_{r-1}}}$.  Finally, we define $Z = Z_0$, and set $\kappa \defeq \sqrt{\frac{d}{2(d+1)}}$.  We define the instance $J' = \tuple{\transition, \restrict I Z, \Gamma'}$ on the nodes $Z$, where $\Gamma' = \tuple{G'_r}_{r \in \N_{>0}}$ is given by $\cinnbhd {G'_r} v = \cinnbhd {G_r} v$ for $v \in Z_r$ and $\cinnbhd {G'_r} v = \set{v, 0}$ for $v \in Z \setminus Z_r$ if $1 \le r \le R$, and any non-split graph for $r > R$.  Note that $\cinnbhd {G_r} v \subseteq Z_{r-1} \subseteq Z$ for all $v \in Z_r$, so $G'_r$ is well-defined on $Z$, and it is non-split since $0$ lies in every in-neighborhood.  In round $r$, a node $v \in Z_r$ then receives the same states as in $J$ and updates to $\state {J'} r v = \kappa^r \apply {g^r} v$.  A node $w \in Z \setminus Z_r$ receives $\set{\state {J'} {r-1} w, 0}$ and must update to the midpoint of these two points (\Cref{lem:symcentroid}), halving its state.  Since $\kappa \ge 1/2$, the norm of its state stays below $\kappa^r$.  Hence the states of $J'$ after round $r$ have radius $\kappa^r$, realized by the nodes in $A_r \subseteq Z_r$.  We therefore have $\roundcontraction {J'} r = \kappa$ in every round $1 \le r \le R$.
\end{proof}

The required number of nodes grows exponentially in $R$.  Specifically, $\cardinality Z \le 2(d+1)^{R+1}$.

\subsection{One-round lower bound}\label{sec:LBOneRound}

In this section, we show a one-round lower bound that approaches $\sqrt{\frac 2 3}$ for $d \to \infty$.  The construction becomes a bit more involved: We split our vector space into two orthogonal subspaces $\values = X \dplus F$ of dimensions $\ell-1$ and $m-1$, form $\ell$ sets $V_i = \Delta_{m} \times \Delta_{\ell-2}$ that are a Cartesian product of a regular $m$-simplex and a regular $(\ell-2)$-simplex, and place these in our vector space such that any two of them intersect in a single point, with all of these intersection points lying on an $(\ell-2)$-sphere in $X$.

\begin{lemma}\label{lem:last}
    Let $\ell,m \in \N_{>0}$ such that $\ell \geq 3$ and $m \geq 1$ with $\ell+m = d+2$.  There exist sets $V_i \subseteq \values$ for $i \in \set{1,\dotsc,\ell}$ such that:
    \begin{itemize}
        \item $V_i$ has a centroid $c_i$ for $i \in \set{1,\dotsc,\ell}$.
        \item The sets $\set{V_i}_{i=1}^\ell$ are pairwise intersecting.
        \item For $M_d' = \bigunion_{i=1}^\ell V_i$ and $M_d = \set{c_i}_{i=1}^\ell$, we have \[\contraction {M_d'} {M_d} = \radius {M_d} / \radius {M_d'} = \lambda = \sqrt{\frac{2}{\gamma(4-\gamma s)}} \mathcomma\] where $\gamma = 1 + \frac 1 m$ and $s = 1 - \frac 1 {\ell-1}$.
    \end{itemize}
    If for every $d \geq 3$ we choose $\ell$ and $m$ dependent on $d$ such that $\ell \xrightarrow{d \to \infty} \infty$ and $m \xrightarrow{d \to \infty} \infty$, we have $\contraction {M_d'} {M_d} \xrightarrow{d\to \infty}\sqrt{\frac 2 3}$.
\end{lemma}

\begin{proof}
    Let $\ell,m \in \N_{>0}$ such that $\ell \geq 3$ and $m \geq 1$ with $\ell+m = d+2$, and set $t = -\frac 1 {\ell-1}$, $s = 1 + t$, $\gamma = 1 + \frac 1 m$, $\lambda = \sqrt{\frac 2 {\gamma(4 - \gamma s)}}$, $\alpha = \frac{\gamma \lambda}{2}$ and $\beta = \lambda - \alpha s$.  Note that $\gamma s < 2$ and $\beta > 0$.

    We now split $\values = X \dplus F$ orthogonally with $\dim X = \ell-1$ and $\dim F = m-1$.  Let $\set{\hat u_i}_{i=1}^\ell \subseteq X$ be a set of unit length vectors that form a regular $(\ell-1)$-simplex with the origin as its centroid.  And similarly, let $\set{f_l}_{l=1}^m \subseteq F$ be a regular $(m-1)$-simplex with the origin as its centroid with $\norm{f_l}^2 = 1 - \frac 1 {m^2}$ (if $m=1$ then $f_1 = 0$).
    
    We now define unit length vectors $v_0^i = - \hat u_i$ and $v_l^i = \frac 1 m \hat u_i + f_l$, and vectors $w_j^i = \hat u_j - t \hat u_i$ of length $\sqrt{1-t^2}$ for $1\le l \le m$, $1\leq i,j \leq \ell$ and $i\neq j$.  For each $i$, the set $A_i = \set{v_l^i}_{l=0}^m$ forms a regular $m$-simplex in $W_i = \spann \set{\hat u_i} \dplus F$ with the origin as its centroid.  Similarly, $B_i = \set{w_j^i \given 1 \le j \le \ell \commaand j \neq i}$ forms a regular $(\ell-2)$-simplex in $\orthcomp{W_i}$ with the origin as its centroid.
    
    Finally, for $1\leq i \leq \ell$, we define $c_i = \lambda\hat u_i$, and define the sets $V_i = c_i + \beta A_i + \alpha B_i$.  Note that $V_i$ has a centroid:  $A_i \subseteq W_i$ and $B_i \subseteq \orthcomp{W_i}$ lie in orthogonal subspaces and each has the origin as its centroid, so their symmetries combine to symmetries of $V_i$ fixing only $c_i$.
    
    These are pairwise intersecting, as \[c_i + \beta v_0^i + \alpha w_j^i = (\lambda - \beta - \alpha t)\hat u_i + \alpha \hat u_j = \gamma \frac{c_i + c_j}{2} = c_i + \beta v_0^j + \alpha w_i^j\] for all $i\neq j$.

    Now let $p_{i,j,l} = c_i + \beta v_l^i + \alpha w_j^i$ for distinct $1\leq i,j\leq \ell$ and $0 \le l \le m$.  Using the fact that $A_i\subseteq W_i$ and $B_i \subseteq \orthcomp{W_i}$ are orthogonal, and that $\hat u_i \in W_i$, for $l \neq 0$, we can calculate
    \begin{align*}
        \norm{p_{i,j,l}}^2
        &= \norm{c_i + \beta v_l^i}^2 + \alpha^2\norm{w^i_j}^2 \\
        &= \lambda^2 + 2\lambda \beta \innprod{\hat u_i}{v_l^i} + \beta^2 + \alpha^2 (1-t^2)  \\
        &= \lambda^2 + \frac{2\lambda \beta}{m} + \beta^2 + \alpha^2 (1-t^2) \mathperiod
    \end{align*}
    Substituting $\beta=\lambda(1-\frac{\gamma s}2)$, $\alpha=\frac{\gamma\lambda}2$, $1-t^2=s(2-s)$ and $\frac 1 m = \gamma-1$:
    \begin{align*}
        \frac{\norm{p_{i,j,l}}^2}{\lambda^2}
         &=1+\frac2m(1-\frac{\gamma s}2)+(1-\frac{\gamma s}2)^2+\frac{\gamma^2}4\,s(2-s)\\
         &=1+\frac2m(1-\frac{\gamma s}2)+(1-\gamma s+\frac{\gamma^2 s^2}4)+(\frac{\gamma^2s}2-\frac{\gamma^2s^2}4)\\
         &=1+(1-\gamma s+\frac{\gamma^2s}2)+(2\gamma-2-\gamma^2s+\gamma s) \\
         &=2\gamma-\frac{\gamma^2s}2=\frac{\gamma(4-\gamma s)}2 \mathcomma
    \end{align*}
    therefore $\norm{p_{i,j,l}} = 1$.  Also note that $\norm{p_{i,j,0}}^2 = \norm{\alpha(\hat u_i + \hat u_j)}^2 = 2\alpha^2s=\frac{\gamma s}{4-\gamma s} < 1$, showing that all points lie in the closed unit ball centered at the origin.  The points $p_{i,j,l}$ with $l \neq 0$ of unit norm average to the origin, so \Cref{lem:maxnorm} shows $\radius {M_d'} = 1$.
    
    With $\norm{c_i} = \lambda$, since they form a regular simplex, we see for $M_d = \set{c_i}_{i=1}^\ell$ that $\radius {M_d} = \lambda$, finally implying $\contraction{M_d'}{M_d} = \lambda$.  Finally, choose $\ell = \ceil{\frac d 2} + 2$ and $m = \floor{\frac d 2}$, so that $\ell + m = d + 2$.  With $a \defeq \frac 1 m$ and $b \defeq \frac 1 {\ell-1}$, we have $b \le a \le \frac 2 {d-1}$, and expanding the denominator gives $\gamma(4 - \gamma s) = 3 + 2a + b + 2ab - a^2 s \le 3 + 5a \le 3 + \frac{10}{d-1}$.  Hence $\lambda^2 \ge \frac{2}{3 + 10/(d-1)} = \frac{2d-2}{3d+7} \xrightarrow{d \to \infty} \frac 2 3$.
\end{proof}

\begin{theorem}
    For every $d \ge 3$ and any transition function, there is an instance in the non-split network model whose one-round contraction is lower bounded by $\sqrt{\frac{2d-2}{3d+7}} \xrightarrow{d \to \infty} \sqrt{\frac 2 3}$.
\end{theorem}

\begin{proof}
        Let $\nodes = \set{\tuple{i,j,l} \given 1 \leq i,j \leq \ell \commaand i \neq j \commaand 1 \leq l \leq m} \union \set{\set{i,j} \given 1 \leq i < j \leq \ell}$.  We define the states $q'_{i,j,l} = p_{i,j,l} = c_i + \beta v_l^i + \alpha w^i_j$ and $q'_{\set{i,j}} = p_{i,j,0} = \alpha(\hat u_i + \hat u_j)$, and, writing $N_i = \set{\tuple{i,j,l} \given j \neq i \commaand 1 \leq l \leq m} \union \set{\set{i,j} \given j \neq i}$, the graph $G$ on $\nodes$ via $\cinnbhd G {\tuple{i,j,l}} = N_i$ for all $i \neq j$ and $\cinnbhd G {\set{i,j}} = N_{\max\set{i,j}}$ for all $i < j$.

        The closed in-neighborhoods contain their own nodes, and $G$ is non-split, since $\set{i,j} \in N_i \intersect N_j$ for all $i \neq j$.  Since $p_{i,j,0} = c_i + \beta v_0^i + \alpha w_j^i = \alpha(\hat u_i + \hat u_j)$, the multiset of states of the nodes in $N_i$ is exactly $V_i$.  Hence every node of group $i$ receives $V_i$ and, by \Cref{lem:symcentroid}, updates to its centroid $c_i$.  The multiset of all states the nodes start with is $M_d'$ with $\radius{M_d'} = 1$, and the multiset of all updated states is $\set{c_i}_{i=1}^\ell = M_d$ with $\radius{M_d} = \lambda \ge \sqrt{\frac{2d-2}{3d+7}} \xrightarrow{d \to \infty} \sqrt{\frac 2 3}$ by \Cref{lem:last}.
\end{proof}

\section{Conclusion}
This paper focused on improving contraction bounds for asymptotic agreement in dynamic networks. We analyzed the \bmp algorithm in the rooted tree and the non-split network models. For the non-split network model, we could show a contraction rate of $\sqrt{2d/(3d+1)}$ for our algorithm, which improves over the previously best-known bound of $\sqrt{7/8}$ via the MidExtremes algorithm. Note that our algorithm can be used to solve approximate agreement. In approximate agreement, the convergence condition is relaxed to $\varepsilon$-agreement: the nodes should terminate with values that lie inside a smallest enclosing ball of radius $\varepsilon$. To solve approximate agreement, it would be necessary to let each node locally keep a round counter. Additionally, the nodes must receive an upper bound on the radius of the smallest enclosing ball of the initial values, and the parameter $\varepsilon$. From this, each node can compute the number of rounds needed until it can terminate. Note that a round counter makes the algorithm no longer memoryless.

On the lower bound side, we present a constructive lower bound of $2/3$ on the contraction rate of any coordinate-free memoryless anonymous algorithm in non-split networks. We also present a better lower bound of $\sqrt{\frac{d}{2(d+1)}}$ on the contraction rate, which assumes that the number of nodes is infinite. For both bounds, our algorithm is not tight. Additionally, we show a lower bound on one-round contraction that can be achieved with a coordinate-free memoryless anonymous algorithm. With respect to this bound, our algorithm is asymptotically tight. Our results show the difficulty in improving the lower bound of $\sqrt{\frac{d}{2(d+1)}}$, as we struggle to construct a self-similar example out of the one-round lower bound. To make the bounds tight, an amortized analysis of the \bmp algorithm could be interesting to consider in the future.

\section*{Acknowledgments}

This work was supported by the German Research Foundation (DFG), Schwerpunktprogramm SPP 2378 - ReNO-2 (grant 511099228), 2025-2029.

The authors of this paper have used artificial intelligence (AI) tools running large language models (LLM) to discover the $2/3$ lower bound from Section~\ref{sec:LBEightDim}, as well as the one-round lower bound of $\sqrt{2/3}$ in Section~\ref{sec:LBOneRound}. For the $2/3$ lower bound in Section~\ref{sec:LBEightDim}, we started with the $\sqrt{1/3}$ lower bound for $d=2$ that we developed ourselves and used an LLM to generalize it to higher dimensions. The LLM came up with constructions only for $d=4$ and $d=8$. It returned that the same construction is harder to generalize because of the Hurwitz theorem. The one-round lower bound of $\sqrt{2/3}$ in Section~\ref{sec:LBOneRound} was a generalization of our  
tightness result from Section~\ref{sec:algAnalysisTightness}. In particular, we asked the LLM to look for an analogous construction that has the necessary symmetry properties, such that any coordinate-free memoryless anonymous deterministic algorithm can only choose one point in every view as the input for the next round. 

Beyond the use of AI to develop the two theoretical results, we used AI to review the manuscript for errors and inconsistencies, consult it for fixing mathematical gaps, and to write the code for our tikz pictures.  The models we used were Kimi K3~\cite{KimiK3}, GLM 5.3~\cite{GLM53}, Claude Opus 5~\cite{ClaudeOpus5}, and Claude Opus 5.5~\cite{ClaudeOpus55}.

We would like to emphasize that the results in this paper (excluding the $2/3$ and the $\sqrt{2/3}$ lower bound) have been developed by the authors. Additionally, the text in this paper has been written entirely by the authors, and the authors have carefully verified the full content of this paper.

\printbibliography

\end{document}